\documentclass[a4paper,12pt]{elsarticle}
\usepackage{amsmath,amssymb,tabularx,hyperref}
\usepackage{xspace,color}

\usepackage{graphicx}
\usepackage{epsfig}

\usepackage{latexsym}

\usepackage{url}

\usepackage{algorithm}
\usepackage[noend]{algpseudocode}
\renewcommand{\algorithmicrequire}{\textbf{Input: }}
\renewcommand{\algorithmicensure}{\textbf{Output: }}

\newtheorem{theorem}{Theorem}[section]
 \newtheorem{corollary}[theorem]{Corollary}
 \newtheorem{lemma}[theorem]{Lemma}
 \newtheorem{proposition}[theorem]{Proposition}
 \newtheorem{definition}[theorem]{Definition}
 \newtheorem{example}[theorem]{Example}
 \newtheorem{remark}[theorem]{Remark}
 \newproof{proof}{Proof}

\newcommand{\vir}[1]{``#1''}
\newcommand{\BWT}{\ensuremath{BWT}\xspace}
\newcommand{\bwt}{\ensuremath{bwt}}
\newcommand{\ABWT}{\ensuremath{ABWT}\xspace}
\newcommand{\abwt}{\ensuremath{abwt}}
\newcommand{\LF}{LF}

\newcommand{\Pisig}{\Pi_\Sigma}
\newcommand{\pid}{Id}
\newcommand{\pre}{Rev}
\newcommand{\rank}{\mathop{\mathrm{rank}}\nolimits}

\newcommand{\plex}{\preceq_{lex}}
\newcommand{\palt}{\preceq_{alt}}
\newcommand{\lalt}{\prec_{alt}}
\newcommand{\lex}{{lex}}
\newcommand{\alt}{{alt}}

\newcommand{\Oh}{{\cal O}}
\def\rle(#1){\mathop{\mathrm{rle}}\nolimits(#1)}

\def\BCR{{\rm BCR}\xspace}

\def\Msort#1{M_{#1}(w)}
\long\def\ignore#1{}

\makeatletter
\let\@@pmod\pmod
\DeclareRobustCommand{\pmod}{\@ifstar\@pmods\@@pmod}
\def\@pmods#1{\mkern4mu({\operator@font mod}\mkern 6mu#1)}
\makeatother

\algnewcommand\KwAnd{\textbf{and }}

\begin{document}

\begin{frontmatter}
\title{The Alternating  BWT: an algorithmic perspective}

%
%
\author[Palermo]{Raffaele Giancarlo}\ead{raffaele.giancarlo@unipa.it}
\author[UEP-Pisa]{Giovanni Manzini}\ead{giovanni.manzini@uniupo.it}
\author[Palermo]{Antonio Restivo}\ead{antonio.restivo@unipa.it}
\author[Pisa]{Giovanna Rosone}\ead{giovanna.rosone@unipi.it}
\author[Palermo]{Marinella Sciortino}\ead{marinella.sciortino@unipa.it}
\address[Palermo]{University of Palermo, Italy}
\address[UEP-Pisa]{University of Eastern Piedmont and IIT-CNR Pisa, Italy}
\address[Pisa]{University of Pisa, Italy}

\begin{abstract}

The Burrows-Wheeler Transform (BWT) is a word transformation introduced in 1994 for Data Compression. It has become a fundamental tool for designing self-indexing data structures, 
with important applications in several area in science and engineering.
The Alternating Burrows-Wheeler Transform (ABWT) is another transformation recently introduced in [Gessel et al. 2012]
and studied in the field of Combinatorics on Words. It is analogous to the BWT, except that it uses an alternating lexicographical order instead of the usual one. 
Building on results in [Giancarlo et al. 2018], where we have shown that BWT and ABWT are part of a larger class of reversible transformations, here we provide a combinatorial and algorithmic study of the novel transform ABWT. We establish a deep analogy between BWT and ABWT by proving they are the only ones in the above mentioned class to be rank-invertible, a novel notion guaranteeing  efficient invertibility. In addition, we show that the backward-search procedure can be efficiently generalized to the ABWT; this result implies that also the ABWT can be used as a basis for efficient compressed full text indices. 
Finally, we prove that the ABWT can be efficiently computed by using a combination of the Difference Cover suffix sorting algorithm  [K\"{a}rkk\"{a}inen et al., 2006] with a linear time algorithm for finding the minimal cyclic rotation of a word with respect to the alternating lexicographical order.

\ignore{Finally, we prove that the ABWT can be efficiently computed by using a variant of the Difference Cover suffix sorting algorithm  [K\"{a}rkk\"{a}inen et al., 2006].  Further improvements are obtained via the use of the combinatorial notion of Galois word, which is the minimal cyclic rotation of a word, with respect to the alternating lexicographical order. A linear algorithm to find the Galois rotation of a word is introduced.}

\end{abstract}


\end{frontmatter}

\section{Introduction}


Michael Burrows and David Wheeler introduced in 1994 a reversible word transformation~\cite{bwt94}, denoted by $\BWT$, that turned out to have \vir{myriad virtues}.  At the time of its introduction in the field of text compression, the Burrows-Wheeler Transform was perceived as a {magic box}: when used as a preprocessing step it would bring rather weak compressors to be competitive in terms of compression ratio with the best ones available~\cite{cj/fenwick96}. 
In the years that followed, many studies have shown the effectiveness of $\BWT$ and its central role in the field of Data Compression due to the fact that it can be seen as a \vir{booster} of the performance of memoryless compressors~\cite{FGMS2005,GIA07,Manzini2001}. Moreover, it was shown in~\cite{Ferragina:2000} that the $\BWT$ can be used to efficiently search for occurrences of patterns inside the original text. Such capabilities of the $\BWT$ have originated the field of Compressed Full-text Self-indices~\cite{books/MBCT2015,books/daglib/0038982}.
The remarkable properties of the $\BWT$ have aroused great interest both from the theoretical and applicative points of view \cite{MantaciRRS07,MantaciRRS08,MantaciRS08,YangZhangWang2010,LiDurbin10,bioinformatics/CoxBJR12,KKbioinf15,RosoneSciortino_CiE2013,tcs/GagieMS17,PPRS2019}.

\ignore{More in detail, the $\BWT$ is defined via a sorting in lexicographic order of all the cyclic rotations of the input word.  The $\BWT$ can be computed in linear time, it produces strings which are provably compressible in terms of the high order entropy of the input, and it can be inverted in linear time by just counting operations (such a property will be formalized in the follows as \emph{rank-invertibility}). Despite its simplicity, the $\BWT$ presents some combinatorial properties that}

In the context of Combinatorics on Words, many studies have addressed the characterization of the words that become the most compressible after the application of the $\BWT$ \cite{MaReSc,puglisiSimpson,PakRedlich2008,RestivoRosoneTCS2009,RestivoRosoneTCS2011,Ferenczi_Zamboni2012arXiv}. Recent studies have focused on measuring the \vir{clustering effect} of $\BWT$, which is a property related to its boosting role as preprocessing of a text compressor \cite{MantaciRRSV17,MantaciRRS17}.

In~\cite{CDP2005}, the authors characterize the $\BWT$ as the inverse of a known bijection between words and multisets of primitive necklaces~\cite{GeRe}.
From this result, in~\cite{GesselRestivoReutenauer2012} the authors introduce and study the basic properties of the {\em Alternating $\BWT$}, $\ABWT$ from now on. It is a transformation on words analogous to the $\BWT$ but the cyclic rotations of the input word are sorted by using the {\em alternating} lexicographic order instead of the usual lexicographic order. The alternating lexicographic order is defined for infinite words as follows: the first letters are compared with the given alphabetic order, in case of equality the second letters are compared with the opposite order, and so on alternating the two orders for even/odd positions. 

In this paper we show that the $\ABWT$ satisfies most of the properties that make the $\BWT$ such a useful transformation. Not only the $\ABWT$ can be computed and inverted in linear time, but also the \emph{backward-search} procedure, which is the basis for indexed exact string matching on the $\BWT$, can be efficiently generalized to the $\ABWT$. This implies that the $\ABWT$ can be used to build an efficient compressed full text index for the transformed string, similarly to the $\BWT$. Note that the variants of the original $\BWT$ which have been introduced so far in the literature~\cite{Schindler1997,dcc/ChapinT98,DAYKIN2017}, were either simple modifications that do not bring new theoretical insight or they were significantly different but without the remarkable compression and search properties of the original~\BWT (see~\cite[Section~2.1]{GMRS_CPM2019arxiv} for a more detailed discussion of these variants).


The existence of the $\ABWT$ shows that the classical lexicographic order is not the only order relation that one can use to obtain a reversible transformation. Indeed, lexicographic and alternating lexicographic order are two particular cases of a more general class of order relations considered in \cite{DOLCE_Reut_Rest2018,Reutenauer2006}.  In a preliminary version of this paper~\cite{GMRRS_DLT2018} we introduce therefore a class of reversible word transformations based on the above order relations that includes both the original $\BWT$ and the Alternating $\BWT$.
Within this class, we introduce the notion of \emph{rank-invertibility}, a property that guarantees that the transformation can be efficiently inverted using rank operations, and we prove that $\BWT$ and $\ABWT$ are the only transformations within this class that are rank-invertible.

We consider also the problem of efficiently computing the $\ABWT$. We first show how to generalize to the alternating lexicographic order the Difference Cover technique introduced in~\cite{Karkkainen:2006}. 
This result leads to the design of time optimal and space efficient algorithms for the construction of the $\ABWT$ in different models of computation when the input string ends with a unique end-of-string symbol. Finally, we explore some combinatorial properties of the \emph{Galois words}, which are minimal cyclic rotations within a conjugacy class, with respect to the alternating lexicographical order. We provide a linear time and space algorithm to find the Galois rotation of a given word and we show that, combining this algorithm with the Difference Cover technique, the $\ABWT$ can be computed in linear time even when the input string does not end with a unique end-of-string symbol.


Motivated by the discovering of the $\ABWT$, in \cite{GMRS_CPM2019arxiv} the authors explore a class of string transformations that includes the one considered in this paper. In this larger class, the cyclic rotations of the input string are sorted using an alphabet ordering that depends on the longest common prefix of the rotations being compared. Somewhat surprisingly some of the transformations in this class do have the same properties of the \BWT and \ABWT, thus showing that our understanding of these transformations is still incomplete. 

\section{Preliminaries}
\label{sec:prel}

Let $\Sigma =\{c_0, c_1, \ldots, c_{\sigma-1}\}$ be an ordered constant size alphabet with $c_0< c_1< \ldots < c_{\sigma-1}$, where $<$ denotes the standard lexicographic order. We denote by $\Sigma^*$ the set of words over $\Sigma$.
Let $w = w_0w_1 \cdots w_{n-1} \in \Sigma^*$ be a finite word, we denote by $|w|$ its length $n$.
We use $\epsilon$ to denote the empty word.
We denote by $|w|_c$ the number of occurrences of a letter $c$ in $w$. The Parikh vector $P_w$ of a word $w$ is a $\sigma$-length array of integers such that for each $c\in \Sigma$, $P_w[c]=|w|_c$. Given a word $x$ and $c\in\Sigma$, we write $\rank_c(x,i)$ to denote the number of occurrences of $c$ in $x[0,i]$.

Given a finite word $w$, a \emph{factor} of $w$ is written as $w[i,j] = w_i \cdots w_j$, with $0\leq i \leq j \leq n-1$. 
A factor of type $w[0,j]$ is called a \emph{prefix}, while a factor of type $w[i,n-1]$ is called a \emph{suffix}. The longest proper factor of $w$ that is both prefix and suffix is called \emph{border}. The $i$-th symbol in $w$ is denoted by $w[i]$. Two words $x,y\in \Sigma^*$ are {\em conjugate}, if $x=uv$ and $y=vu$, where $u,v\in \Sigma^*$. We also say that $x$ is a {\em cyclic rotation} of $y$. A word $x$ is {\em primitive} if all its cyclic rotations are distinct. A primitive word is a \emph{Lyndon word} if it is smaller than all of its conjugates. Conjugacy between words is an equivalence relation over $\Sigma^*$. A word $z$ is called a {\em circular factor} of $x$ if it is a factor of some conjugate of $x$. 

Given two words of the same length $x=x_0x_1\ldots x_{s-1}$ and $y=y_0y_1\ldots y_{s-1}$, we write $x \plex y$ if and only if  $x=y$ or $x_i<y_i$, where $i$ is the smallest index in which the corresponding characters of the two words differ. Analogously, and with the same notation as before, we write $x \palt y$ if and only if $x=y$ or (a) $i$  is even and $x_i<y_i$ or (b) $i$ is odd and $x_i>y_i$. Notice that $\plex$ is the standard lexicographic order relation  on words while $\palt$ is the \emph{alternating} lexicographic order relation. Such orders are used in Section \ref{sec:BWT_ABWT} to define two different transformations on words.

The {\em run-length encoding} of a word $w$, denoted by $\rle(w)$, is a sequence of pairs
$(w_i, l_i$) such that $w_iw_{i+1}\cdots w_{i+l_i-1}$ is a maximal run of a letter $w_i$ (i.e., $w_i=w_{i+1}=\cdots =w_{i+l_i-1}$, $w_{i-1} \neq w_i$
and $w_{i+l_i} \neq w_i$), and all such maximal runs are listed in $\rle(w)$ in the order they
appear in $w$. We denote by $\rho(w)=|\rle(w)|$ i.e., is the number of pairs in $w$, or equivalently the number of equal-letter runs in $w$.
Moreover we denote by $\rho(w)_{c_i}$ the number of pairs $(w_j, l_j)$ in $\rle(w)$ where $w_j=c_i$.
Notice that $\rho(w) \leq \rho(w_1) + \rho(w_2) +\cdots + \rho(w_p)$, where $w_1w_2 \cdots w_p = w$ is any partition of $w$.

The zero-th order empirical entropy of the word $w$ is defined as
$$
H_0(w) = - \sum_{i=0}^{\sigma-1}\frac{|w|_{c_i}}{|w|}\log \frac{|w|_{c_i}}{|w|}
$$
(all logarithms are taken to the base $2$ and we assume $0\log 0 = 0$). The value $|w|H_0(w)$ is the output size of an ideal compressor that uses $-\log ({|w|_{c_i}}/{|w|})$ bits to encode each occurrence of symbol $c_i$. This is the minimum size we can achieve using a uniquely decodable code in which a fixed codeword is assigned to each symbol.

For any length-$k$ factor $x$ of $w$, we denote by $x_w$ the sequence of characters preceding
the occurrences of $x$ in $w$, taken from left to right. If $x$ is not a factor of $w$ the word $x_w$ is empty. 
The $k$-th order empirical entropy of $w$ is defined as
$$
H_k(w) = \frac{1}{|w|}\sum_{x\in\Sigma^k} |x_w|H_0(x_w).
$$
The value $|w|H_k(w)$ is a lower bound to the output size of any compressor that encodes each symbol with a code that only depends on the symbol itself and on the $k$ preceding symbols. Since the use of a longer context helps compression, it is not surprising that for any $k\geq 0$ it is $H_{k+1}(w) \leq H_{k}(w)$.

\section{BWT and Alternating BWT}\label{sec:BWT_ABWT}

In this section we describe two different invertible transformations on words based on the lexicographic and alternating lexicographic order, respectively. 
Given a primitive word $w$ of length $n$ in  $\Sigma^*$, \emph{the Burrows-Wheeler transform}, denoted by $\BWT$ \cite{bwt94} 
and \emph{the Alternating Burrows-Wheeler transform}, denoted by $\ABWT$ \cite{GesselRestivoReutenauer2012}
for $w$ are defined constructively as follows:
\begin{enumerate}
  \item \label{defStepA} Create the matrix $M(w)$ of the cyclic rotations of $w$;
  \item \label{defStepB} Create the matrix 
  \begin{enumerate} 
  \item for $\BWT$, $M_{\lex}(w)$ by sorting the rows of $M(w)$ according to $\plex$;
  \item for $\ABWT$, $\Msort{\alt}$ by sorting the rows of $M(w)$ according to $\palt$;
  \end{enumerate}
  \item \label{defStepC} Return as output the pair 
  \begin{enumerate} 
  \item for $\BWT$, $(\bwt(w), I)$, where $\bwt(w)$ is the last column $L$ in the matrix $\Msort{\lex}$ 
  \item for $\ABWT$, $(\abwt(w), I)$ where $\abwt(w)$ is the last column $L$ in the matrix $\Msort{\alt}$
  \end{enumerate}
  and, in both case, the integer $I$ giving the position of $w$ in that matrix. 
\end{enumerate}

An example of the above process, together with the corresponding output, is provided in Fig. \ref{fig:bwt}.

\begin{figure}[ht]
{\small
$$\arraycolsep=2.5pt
\begin{array}{cccccc}
            &            &           &            &            &            \\          
             &            &            &            &            &            \\
          a &          b &          r &          a &          c &          a \\
          b &          r &          a &          c &          a &          a \\
          r &          a &          c &          a &          a &          b \\
          a &          c &          a &          a &          b &          r \\
          c &          a &          a &          b &          r &          a \\
          a &          a &          b &          r &          a &          c \\
            &            &            &            &            &            \\
           \multicolumn{6}{c}{M(w)}            \\ 
\end{array}
\qquad\qquad
\begin{array}{ccccccccc}
    &                  & F &  &   &   &   &  L  \\
    &                  &  \downarrow &&  &  & &\downarrow  \\
    &                  & a & a & b & r & a & c    \\
    &                  & a & b & r & a & c & a    \\
  I & \rightarrow\;      & a & c & a & a & b & r    \\
    &                  & b & r & a & c & a & a    \\
    &                  & c & a & a & b & r & a    \\
    &                  & r & a & c & a & a & b  \\
    &                  &   &   &    &   &  &     \\
    &                  & \multicolumn{6}{c}{\Msort{\lex}}
\end{array}
\qquad\qquad
\begin{array}{ccccccccc}
    &                 & F &   &   &   &  & L  \\
    &                 & \downarrow &   &   &   &  & \downarrow  \\
  I & \rightarrow\;     & a & c & a & a & b & r    \\
    &                 & a & b & r & a & c & a    \\
    &                 & a & a & b & r & a & c    \\
    &                 & b & r & a & c & a & a    \\
    &                 & c & a & a & b & r & a    \\
    &                 & r & a & c & a & a & b  \\
    &                 &   &   &    &   &   &   \\
    &                 & \multicolumn{6}{c}{\Msort{\alt}}
\end{array}
$$
}
\caption{Left: the matrix $M(w)$ of all cyclic rotations of the word $w = acaabr$. Center:  the matrix $\Msort{\lex}$; the pair $(caraab,2)$ is the output $\bwt(w)$. Right: the matrix $\Msort{\alt}$; the pair $(racaab,0)$ is the output of $\ABWT(w)$.} \label{fig:bwt}
\end{figure}

\begin{remark}\label{rem:conjugacy}
If two words are conjugate the $\BWT$ (resp. $\ABWT$) will have the same column $L$ and differ only in $I$, whose purpose is only to distinguish between the different members of the conjugacy class. However, $I$ is not necessary in order to recover the matrix $M$ from the last column $L$.
\end{remark}

The following proposition, proved in \cite{GMRRS_DLT2018}, states that three well known properties of the $\BWT$ hold, in a slightly modified form, for the $\ABWT$ as well.  
Here we report the proof for the sake of completeness.

\begin{proposition}\label{p-fondBWT}
Let $w$ be a word and let $(L,I)$ be the output of $\BWT$ or $\ABWT$ applied to $w$. The following properties hold:
\begin{enumerate}
\item\label{proper0} Let $F$ denote the first column of $\Msort{\lex}$ (resp. $\Msort{\alt})$, then $F$ is obtained by lexicographically sorting the symbols of $L$.
\item\label{proper1} For every  $i$, $0 \leq i < n$, $L[i]$ circularly precedes $F[i]$ in the original word, for both $\BWT$ and $\ABWT$.
\item\label{proper2} For each symbol $a$, and $1 \leq j \leq |w|_a$, the $j$-th occurrence of $a$ in $F$ corresponds
\begin{enumerate}
\item\label{proper2_lex} for $\BWT$, to its $j$-th occurrence  in $L$
\item\label{proper2_alt} for $\ABWT$, to its $(|w|_a-j+1)$-th occurrence in $L$.
\end{enumerate}
\end{enumerate}
\end{proposition}

\begin{proof}
Properties \ref{proper0}, \ref{proper1} and \ref{proper2_lex} for the $\BWT$ have been established in~\cite{bwt94}. Properties \ref{proper0} and \ref{proper1} for the $\ABWT$ are straightforward. To prove property \ref{proper2_alt}, consider two rows $i$ and $j$ in $\Msort{\alt}$ with $i < j$ starting with the symbol $a$. Let $w_i$ and $w_j$ be the two conjugates of $w$ in rows $i$ and $j$ of $\Msort{\alt}$. By construction we have $w_i = a u$,  $w_j = a v$ and $w_i \palt w_j$. To establish Property~\ref{proper2_alt}, we need to show that row $w_j$ cyclically rotated precedes in the $\palt$ order row $w_i$ cyclically rotated. 
In other words, we need to show that
$$
au \palt av \;\Longrightarrow\; va \palt ua.
$$
To prove the above implication, we notice that if the first position in which $au$ and $av$ differ is odd (resp. even) then the first position in which $va$ and $ua$ differ will be in an even (resp. odd) position. The thesis follow by the alternate use of the standard and reverse order in $\palt$ (see \cite{GesselRestivoReutenauer2012} for a different proof of the same property).\qed
\end{proof}

It is well known that in the $\BWT$ the occurrences of the same symbol appear in columns $F$ and $L$ in the \textbf{same} relative order; according to Property~\ref{proper2_alt}.
In the $\ABWT$, the occurrences in $L$ appear in the \textbf{reverse} order than in $F$. For example, in Fig.~\ref{fig:bwt} (right) we see that the $a$'s of $acaabr$ in the columns $F$ appear in the order 1st, 3rd, and 2nd, while in column $L$ they are in the reverse order: 2nd, 3rd, and 1st. 

Proposition \ref{p-fondBWT} is the key motivations to efficiently recover the original string from the output of $\BWT$ or $\ABWT$, as we will see in Section \ref{sec:uniq}.

Note that, although $\BWT$ and $\ABWT$ are very similarly defined, they are very different combinatorial tools. Combinatorial aspects that distinguish $\ABWT$ and $\BWT$ can be found  in~\cite{GesselRestivoReutenauer2012,GMRRS_DLT2018}, which makes it interesting to study $ABWT$ in terms of tool characterizing families of words. 

However, in \cite{GMRRS_DLT2018} we experimentally tested $\ABWT$ as pre-processing of a compression tool, by comparing its performance with a $BWT$-based compressor. We have shown that the behaviour of the two transformations is essentially equivalent in terms of compression. Actually, such experiments confirm a theoretical result we proved in \cite{GMRRS_DLT2018} for a larger class of transformations that can be seen as a generalization of the $\BWT$ and that includes the $\ABWT$ as a special case. In next section, we give a brief description of the properties we proved in \cite{GMRRS_DLT2018} for such a class of transformations, all of which also hold for the $\ABWT$. 

\ignore{For instance, $\BWT$ allows to characterize a family of words very well known in the field of Combinatorics in Words, Standard Sturmian words \cite{Loth2}. These words have several characterizations as, for instance, a special decomposition into palindrome words and an extremal property on the periods of the word that is closely related to Fine and Wilf's theorem \cite{deLuca1997,deLucaMignosi1994}. Moreover, they also appear as extremal case in the Knuth-Morris-Pratt pattern matching algorithm (see \cite{KnuthMorrisPratt1977}). It has been proved \cite{MaReSc} that, for binary alphabets, standard Sturmian words represent the extremal case of $\BWT$ in the sense that the transformation produces a total clustering of all the instances of any character. Thus, in terms of number of runs, $\rho(\bwt(w))=2$ if and only if $w$ is a conjugate of standard Sturmian words. 
The same property does not hold for the $\ABWT$. For example, for $w=abaababa$,  it is $\bwt(w)=bbbaaaaa$ and $\abwt(w)=ababbaaa$. More in general, one can prove that for every not unary word $w$ having length greater that 2, it is $\rho(\abwt(w))>2$. 


Other combinatorial aspects that distinguish $\ABWT$ and $\BWT$ have been studied in~\cite{GesselRestivoReutenauer2012}. }

\section{Generalized BWTs: a synopsis} \label{sec:abwt}

In this section we describe the class of Generalized BWTs, introduced in \cite{GMRRS_DLT2018}, by reporting their main properties.

Given the alphabet $\Sigma$ of size $\sigma$, in the following, we denote by $\Pisig$ the set of $\sigma!$ permutations of the alphabet symbols. 
Two important permutations are distinguished in $\Pisig$:
the identity permutation $\pid$ corresponding to the lexicographic order, and the reverse permutation $\pre$ corresponding to the reverse lexicographic order. We consider generalized lexicographic orders introduced in \cite{Reutenauer2006} (cf. also \cite{DOLCE_Reut_Rest2018}) that, for the purposes of this paper, can be formalized as follows.

\begin{definition}\label{def_perm}
Given a $k$-tuple $K = (\pi_0,\pi_1,\ldots,\pi_{k-1})$ of elements of $\Pisig$, we denote by $\preceq_K$ the lexicographic order such that given two words of the same length $x =x_0x_1\cdots x_{s-1}$ and $y = y_0 y_1 \cdots y_{s-1}$ it is $x \preceq_K y$ if and only if $x=y$ or $x_i <_i y_i$ where $i$ is the smallest index such that $x_i \neq y_i$, and  $<_i$ is the lexicographic order induced by the permutation $\pi_{i\bmod k}$. Without loss of generality, we can assume $\pi_0 = \pid$.
\end{definition}

Using the above definition, a class of generalized $\BWT$s can be defined as follows:

\begin{definition}
Given a $k$-tuple $K =(\pid,\pi_1,\ldots,\pi_{k-1})$ of elements of $\Pisig$, we denote by $\BWT_K$ the transformation mapping a primitive word $w$ to the last column $L$ of the matrix $M_K(w)$ containing the cyclic rotations of $w$ sorted according to the lexicographic order $\preceq_K$.
The output of $\BWT_K$ applied to $w$ is the pair $(\bwt_K(w),I)$, where $\bwt_K(w)$ is the last column $L$ of the matrix and $I$ is the row of $M_K(w)$ containing the word $w$.
\end{definition}

Note that for $K=(\pid)$, $\BWT_K$ is the usual $\BWT$, while for $K =(\pid,\pre)$, $\BWT_K$ coincides with the $\ABWT$ defined in Section~\ref{sec:BWT_ABWT}.

\begin{remark}\label{rem:time}
For most applications, it is assumed that the last symbol of $w$ is a unique end-of-string marker smaller than each symbol of the alphabet $\Sigma$. Under this assumption, lexicographically sorting $w$'s cyclic rotations
is equivalent to building the suffix tree~\cite{Lothaire:2005,Gusfield1997} for $w$, which can be done in linear time.  In this setting, we can compute $\bwt_K(w)$ in linear time: we do a depth-first visit of the suffix tree in which the children of each node are visited in the order induced by $K$. In other words, the children of each node $v$ are visited according to the order $\pi_{|v| \bmod k}$ where $|v|$ is the string-depth\footnote{The number of letters in the word obtained by concatenating the labels of the edges in the path from the root of the suffix tree to the node $v$} of node~$v$. Since the suffix tree has $O(|w|)$ nodes, for a constant alphabet the whole procedure takes linear time.
\end{remark}

The next result proved in \cite{GMRRS_DLT2018} guarantees that the transformations $\BWT_K$ are invertible. As we specify in next section, the inversion procedure for $ABWT$ is more efficient.


\begin{theorem}
For every $k$-tuple $K=(\pid,\pi_1,\ldots,\pi_{k-1})$ the transformation $\BWT_K$ is invertible in $O(n^3)$ time, where $n=|w|$.\qed
\end{theorem}

Note that recently~\cite{GMRS_CPM2019arxiv}, the complexity for the inversion of a generic transformation in $BWT_K$ has been improved to $\Oh(n^2)$ time.

The following theorem proved in \cite{GMRRS_DLT2018} shows that each transformation $BWT_K$ produces a number of equal-letter runs that is at most the double of the number of equal-letter runs of the input word. This fact generalizes a result proved for $\BWT$ \cite{MantaciRRSV17}.

\begin{theorem}\label{th-upperb_runs}
Given a $k$-tuple $K =(\pid,\pi_1,\ldots,\pi_{k-1})$ and a word $w$ over a finite alphabet $\Sigma$, then 
$$\rho(bwt_K(w))\leq 2\rho(w).$$\qed
\end{theorem}

A key property of $\BWT$ is that it allows to reduce the problem of compressing a string $w$ up to its $r$-th order entropy to the problem of compressing a collection of factors of $\bwt(w^R)$ up to their $0$-th order entropy, where $w^R$ is the reverse of the word $w$. This means that a $\BWT$-based compressor combining $\BWT$ with a zero order (memoryless) compressor, is able to achieve the same high order compression typical of more complex tools such as Lempel-Ziv encoders. In \cite{GMRRS_DLT2018}, we prove that a similar result also holds for the transformation $BWT_K$.

\begin{theorem}\label{Th:entropy}
Let $K$ be a $k$-tuple and $u=\bwt_K(w^{R})$, where $w^R$ is the reverse of the word $w$. For each positive integer $r$, there exists a factorization of $u=u_1u_2\ldots u_m$ such that 
$$
H_r(w)=\frac{1}{|u|}\sum_{i=1}^{m} |u_i|H_0(u_i).
$$
\end{theorem}

\section{Rank-invertible transformations}\label{sec:uniq}

It is well known that the key to efficiently compute the inverse of original \BWT is the existence of a easy-to-compute permutation mapping, in the matrix $\Msort{lex}$, a row index $i$ to the row index $\LF(i)$ containing row $i$ right-shifted by one position.  This permutation is called $LF$-mapping since, by Proposition~\ref{p-fondBWT}, $\LF(i)$ is the position in the first column $F$ of $\Msort{lex}$ corresponding to the $i$-th entry in column $L$: in other words, $F[LF(i)]$ is the same symbol in $w$ as $L[i]$. Again, by Proposition~\ref{p-fondBWT} we have that $L[\LF(i)]$ is the symbol preceding $L[i]$ in the input word $w$. Define $\LF^0(x)=x$ and $\LF^{j+1}(x) = \LF(\LF^j(x))$. If $\bwt(w)=(L,I)$ with $|w|=n$, then by construction $L[I]=w_{n-1}$ and we can recover $w$ with the formula:
\begin{equation}\label{equ:recover}
w_{n-1-j} = L[LF^j(I)]
\end{equation}
Note that the inversion formula~\eqref{equ:recover} only depends on Properties~\ref{proper0} and~\ref{proper1} of Proposition~\ref{p-fondBWT}. Since such properties hold for every generalized transformation $\BWT_K$, \eqref{equ:recover} provides an inversion formula for every transformation in that class. In other words, inverting a generalized $\BWT$ amounts to computing $n$ iterations of the $LF$-mapping.

By Property~\ref{proper2_lex} in Proposition~\ref{p-fondBWT}, the $LF$-mapping for the original $\BWT$ can be expressed using the Parikh vector $P_L$ of $L$ and a rank operation over~$L$:
\begin{equation}\label{eq:bwtlf}
\LF(i) = \sum_{c\in\Sigma}^{c<L(i)} P_L[c] \;+\; \rank_{L[i]}(L,i-1)
\end{equation}
Note that $\sum_{c\in\Sigma}^{c<L(i)} P_L[c]$ is simply the total number of occurrences of symbols smaller than $L[i]$ in $L$, and $\rank_{L[i]}(L,i-1)$ is the number of occurrences of symbol $L[i]$ in among the first $i$ symbols of~$L$.

By Property~\ref{proper2_alt} in Proposition~\ref{p-fondBWT}, for the \ABWT, the corresponding formula is:
\begin{equation}\label{eq:abwtlf}
\LF(i) = \sum_{c\in\Sigma}^{c\leq L(i)} P_L[c] \;-\; \rank_{L[i]}(L,i-1) - 1
\end{equation}
Since the rank operation on (compressed) arrays over finite alphabet can be computed in constant time~\cite{BNtalg14} and the partial sums $\sum_{c<i} P_L[c]$ can be precomputed, the computation of the $LF$ map for both the $\BWT$ and $\ABWT$ takes $O(1)$ time.  This implies that, thanks to the simple structure of its $LF$-mapping, also the $\ABWT$ can be inverted in linear time. 

The computation of the $LF$ map is the main operation also for the so-called {\em backward-search} procedure which makes it possible to use (a compressed version of) $bwt(w)$ as a full text index for $w$~\cite{Ferragina:2005}. The following proposition is the key to  generalize the backward search procedure to the $\ABWT$.  

\begin{proposition}\label{lemma:back}
Given a string $p\in \Sigma^*$, let $[b,e]$ denote the range of rows of $\Msort{alt}$ which are prefixed by $p$. For any $x\in\Sigma$, let
$$
b' = \sum_{c\in\Sigma}^{c\leq x} P_L[c] \;-\; \rank_{x}(L,e-1) -1 \qquad
e' = \sum_{c\in\Sigma}^{c\leq x} P_L[c] \;-\; \rank_{x}(L,b).
$$
If $b'\leq e'$, then $[b',e']$ is the range of rows of $\Msort{alt}$ which are prefixed by $xp$ if $b'>e$ then no rows of $\Msort{alt}$ are prefixed by $xp$ and therefore $xp$ is not a (circular) substring of $w$.
\end{proposition}

\begin{proof}
Assume first $b'\leq e'$. It is immediate that if $i$, $j$ are the positions of the first and last $x$ in $L[b,e]$, then $b' = LF(j)$ and $e'=LF(i)$ and every other $x$ in  $L[b,e]$ is mapped to a position between $b'$ and $e'$. The thesis follows since all rows in $\Msort{alt}$ are rotations of $w$. If $b'>e$, then  $\rank_{x}(L,e-1)=\rank_{x}(L,b)$ and there are no $x$'s in $L[e,b]$ and $xp$ is not a circular substring of $w$.\qed
\end{proof}

Proposition~\ref{lemma:back} implies that if we use a compressed representation of the last columns $L$ of $\Msort{alt}$ supporting constant time rank operations, then, for any pattern $p$, we can compute in $O(|p|)$ time the range of rows of the matrix $\Msort{lex}$ which are prefixed by $p$. Hence, the $\ABWT$ can be used as a compressed index in the same way as the $\BWT$. 

The above results suggest that it is worthwhile to search for other transformations in the class $\BWT_K$ which share the same properties of $\BWT$ and $\ABWT$. Because of the important role played by the rank operation, we introduce the notion of rank-invertibility for the class of $\BWT_K$ transformations.

\begin{definition}
The transformation $\BWT_K$ is rank-invertible if there exists a function~$f_K$ such that, for any word $w$, setting $L=\bwt_K(w)$ we have 
$$
\LF(i) = f_K(P_L,L[i],\rank_{L[i]}(L,i)).
$$
In other words, $LF(i)$ only depends on the Parikh vector $P_L$ of $L$, the symbol $L[i]$, and the number of occurrences of $L[i]$ in $L$ up to position $i$.\qed
\end{definition}
Note that we pose no limit to the complexity of the function $f_K$, we only ask that it can be computed using only $P_L$ and the number of occurrences of $L[i]$ in $L[0,i]$.

We observed that, for $K=(\pid,\pre)$, $\BWT_K$ coincides with \ABWT and it is therefore rank-invertible. The main result of this section is Theorem \ref{theo:mainRI} establishing that \BWT and \ABWT are the only rank-invertible transformations in the class $\BWT_K$. We start our analysis considering the case $|K|=2$.

\ignore{
We observed that, for $K=(\pid,\pre)$, $\BWT_K$ coincides with \ABWT and it is therefore rank-invertible. The main result of this section is to show that \BWT and \ABWT are the only rank-invertible transformations in the class $\BWT_K$.

In the proofs of the following statements we distinguish the case of the words on binary alphabets and the words on alphabets with cardinality greater than $2$. This depends on the fact that in the binary case the only possible permutations on binary alphabet are the identity and reverse permutation. We first consider the case in which $|K|=2$.  Lemma \ref{lemma:t1t2} provides a necessary condition for $\BWT_K$ to be rank-invertible for ternary alphabets. 
}

\begin{lemma}\label{lemma:t1t2}
Let $\Sigma = \{a,b,c\}$, and $K = (\pid,\pi)$, where $\pi$ is a permutation of~$\Sigma$. If there exist two pairs $t_1=(x,y)$ and $t_2=(z,w)$ of symbols of $\Sigma$ such that 
$$
x <_{\pid} y,\qquad  z<_{\pid} w \qquad\mbox{  and } \qquad x <_{\pi} y,\qquad  z>_{\pi} w,
$$
then $\BWT_K$ is not rank-invertible.
\end{lemma}

\begin{proof}
Consider for example the case $\pi = (c,a,b)$. 
Two pairs satisfying the hypothesis are $t_1 = (a,b)$ and $t_2 = (b,c)$ since according to the ordering $<_\pi$ it is
$$
a <_\pi b\qquad\mbox{ and } \qquad b >_\pi c.
$$

Consider now the two words $s_1=aabcc$ and $s_2=abacc$. Both words contain two $a$'s. In the first word the $a$'s are followed respectively by $a,b$ (the symbols in $t_1$), and in $s_2$ the $a$'s are followed by $b,c$ (the symbols in $t_2$).

Let $F_1$, $L_1$ (resp. $F_2$, $L_2$) denote the first and last column of the matrix $M_K$ associated to $\bwt_K(s_1)$ (resp. $\bwt_K(s_2)$). By definition, each matrix is obtained sorting the cyclic rotations of $s_1$ and $s_2$ according to the lexicographic order  $\prec_K$ where symbols in odd positions are sorted according to the usual alphabetic order, while symbols in even positions are sorted according to the ordering $\pi$. We show the two matrices in Fig.~\ref{fig:s1s2}, where we use subscripts to distinguish the two $a$'s occurrences in $s_1$ and $s_2$.

The relative position of the two $a$'s in $L_1$ is determined by the symbols following them in $s_1$, namely those in $t_1=(a,b)$. Since these symbols are in the first column of the cyclic rotations matrix, which is sorted according to the usual alphabetic order, the two $a$'s appear in $L_1$ in the order $a_1, a_2$. The same is true for $L_2$: since the pair $t_2$ is also sorted, the two $a$'s appear in $L_2$ in the order $a_1, a_2$.

The position of the two $a$'s in $F_1$ is also determined by the symbols following them in $s_1$; but since these symbols are now in the second column, their relative order is determined by the ordering $\pi$. Hence the two $a$'s appear in $F_1$ in the order $a_1, a_2$. In $F_2$ the ordering of the $a$'s is $a_2, a_1$ since it depends on the $\pi$-ordering of the symbols of $t_2$ which {\em by construction} is different than their $\pid$-ordering.

Note that $s_1$ and $s_2$ have the same Parikh vector $\langle 2,1,2\rangle$. If, by contradiction, $\BWT_K$ were rank invertible, the function $f_K$ should give the correct LF-mapping for both $s_1$ and $s_2$. 
This is impossible since for $s_1$ we should have
$$
f_K(\langle 2,1,2\rangle,a,1) = 1,\qquad
f_K(\langle 2,1,2\rangle,a,2)=2,
$$
while for $s_2$ we should have
$$
f_K(\langle 2,1,2\rangle,a,1) = 2,\qquad
f_K(\langle 2,1,2\rangle,a,2)=1.
$$

In the general case of an arbitrary permutation $\pi$ satisfying the hypothesis of the lemma the reasoning is the same. Note that such permutations are $(a,c,b)$, $(b,a,c)$, $(b,c,a)$ and $(c,a,b)$. Given the two pairs $t_1$ and $t_2$ we build two words $s_1$ and $s_2$ with Parikh vector $\langle 2,1,2\rangle$ such that in $s_1$ (resp. $s_2$) the two occurrences of $a$ are followed by the symbols in $t_1$ (resp. $t_2$). We then build the rotation matrices as before, and we find that in both $L_1$ and $L_2$ the two $a$'s are in the order $a_1, a_2$. However, in columns $F_1$ and $F_2$ the two $a$'s are not in the same relative order since it depends on the ordering $\pi$, and, by construction, such an order is not the same. Reasoning as before, we get that there cannot exist a function $f_K$ giving the correct LF-mapping for both $s_1$ and $s_2$.\qed
\end{proof}

\begin{figure}[tb]
{
$$\arraycolsep=4.5pt
\begin{array}{llllll}
 & F_1 &             &      &    & L_1 \\
 &\downarrow&        &      &    &\downarrow\\s_1 \rightarrow
 &   a_1&  a_2       &   b&  c &   c \\ 
 &   a_2&  b         &   c&  c &   a_1 \\
 &   b &    c        &    c &  a_1 &  a_2 \\
 &   c &        c &    a_1 & a_2&  b \\
 &   c &      a_1&   a_2&  b &  c 
\end{array}
\qquad\qquad
\begin{array}{llllll}
 & F_2 &     &   &  & L_2  \\
 & \downarrow    &   &   &  & \downarrow  \\
 & a_2 & c & c & a_1 & b\\s_2 \rightarrow
 & a_1 & b & a_2 & c & c\\ 
 & b   & a_2 & c & c & a_1\\
 & c   & c & a_1 & b & a_2\\
 & c   & a_1 & b & a_2 & c
\end{array}
$$
}
\caption{Cyclic rotation matrices for the words $s_1$ and $s_2$. 
We use subscripts to distinguish the two occurrences of $a$ in each word.}\label{fig:s1s2}
\end{figure}

\begin{lemma}\label{lemma:ix}
Let $|\Sigma|\geq 2$ and $K=(\pid,\pi)$. Then $\BWT_K$ is rank-invertible if and only if  $\pi=\pid$ or $\pi=\pre$.
\end{lemma}

\begin{proof}
If $|\Sigma|=2$ the result is trivial since the only possible permutations on binary alphabet are the identity and reverse permutation. Let us assume $|\Sigma|\geq 3$. We need to prove that if $\pi\neq\pid$ and $\pi\neq\pre$ then $\BWT_K$ is not rank-invertible. 

Note that any permutation $\pi$ over the alphabet $\Sigma$ induces a new ordering on any triplet of symbols in $\Sigma$. For example, if $\Sigma = \{ a, b, c, d, e, f\}$ the permutation $\pi = (d,e,c,f,a,b)$ induces on the triplet $\{a,b,c\}$ the ordering $\pi_{abc} = (c,a,b)$. It is easy to prove by induction on the alphabet size that, if $\pi\neq\pid$ and $\pi\neq\pre$, then there exists a triplet $\{x,y,z\}$, with $x<y<z$, such that $\pi_{xyz} \neq (x,y,z)$ and $\pi_{xyz} \neq (z,y,x)$.
That is, $\pi$ restricted to $\{x,y,z\}$ is different from the identity and reverse permutation. Without loss of generality we can assume that the triplet is $\{a,b,c\}$.

For any permutation $\pi_{abc}$, different from $(a,b,c)$ and $(c,b,a)$, there exist two pairs of symbols satisfying the hypothesis of Lemma~\ref{lemma:t1t2}. Hence, we can build two words $s_1$ and $s_2$ which show that $\BWT_K$ is not rank-invertible. Note that the argument in the proof of Lemma~\ref{lemma:t1t2} is still valid if we add to $s_1$ and $s_2$ the same number of occurrences of symbols in $\Sigma$ different from $a,b,c$ so that $s_1$ and $s_2$ are effectively over an alphabet of size $|\Sigma|$.\qed
\end{proof}

Lemma~\ref{lemma:ix} establishes which $\BWT_K$ transformations are rank-invertible when $|K|=2$. To study the general case $|K|>2$, we start by establishing a simple corollary.

\begin{corollary}\label{cor:ix*}
Let $|\Sigma|\geq 3$ and $K=(\pid,\pi, \pi_2,\ldots, \pi_{k-1})$. If $\pi\neq\pid$ and $\pi\neq\pre$ then $\BWT_K$ is not rank-invertible.
\end{corollary}

\begin{proof}
We reason as in the proof of Lemma~\ref{lemma:ix}, observing that the presence of the permutations $\pi_2, \ldots,\pi_{k-1}$ has no influence on the proof since the row ordering is determined by the first two symbols of each rotation.\qed
\end{proof}

The following three lemmas establish necessary conditions on the structure of the tuple $K$ for $\BWT_K$ to be rank-invertible. In particular, the following lemma shows that $\BWT_K$ is not rank-invertible if $K$ contains anywhere a triplet $(\pid,\pid,\pi)$ with $\pi\neq\pid$.

\begin{lemma}\label{lem:*ii*}
Let $|\Sigma|\geq 2$ and $K=(\pid,\pi_1,\ldots,\pi_{i-1},\pid,\pid,\pi, \pi_{i+3},\ldots,\pi_{k-1})$, $i\geq 0$, with $\pi\neq\pid$. 
Then $\BWT_K$ is not rank-invertible.
\end{lemma}

\begin{proof}
Note that when $i=0$, the $k$-tuple $K$ starts with the triplet $(\pid,\pid,\pi)$. We first analyze the case $|\Sigma|=2$ implying  that $\pi=\pre$. Let us consider the words
$s_1=a_1b^ia_2b^{i+1}bb$ and $s_2=a_1b^{i+1}a_2b^{i+1}b$, where we use subscripts to distinguish the two different occurrences of the symbol $a$. It is easy to see that, in the cyclic rotations matrix for $s_1$, $a_1$ precedes $a_2$ in both the first and the last column. 
Hence if $\BWT_K$ were rank-invertible we should have
$$
f_K(\langle 2,2i+3\rangle,a,1) = 1,\qquad
f_K(\langle 2,2i+3\rangle,a,2) = 2.
$$
At the same time, in the cyclic rotations matrix for $s_2$, $a_1$ precedes $a_2$ in the last columns, but in the first column $a_2$ precedes $a_1$ since the two rotations prefixed by $a$ differ in the third column and $b <_{\pre} a$. Therefore we should have
$$
f_K(\langle 2,2i+3\rangle,a,1) = 2,\qquad
f_K(\langle 2,2i+3\rangle,a,2) = 1.
$$
Hence $\BWT_K$ cannot be rank-invertible.

Let us consider the case $|\Sigma|\geq 3$. Since $\pi\neq\pid$ there are two symbols, say $b$ and $c$, such that their relative order according to $\pi$ is reversed, that is, $b<c$ and $c <_\pi b$. Consider now the words $s_1= a_1c^iba_2c^iccc$ and $s_2=a_1c^{i+1}ba_2c^{i+1}c$  where we use subscripts to distinguish the two different occurrences of the symbol $a$. It is immediate to see that, in the cyclic rotations matrix for $s_1$, $a_1$ precedes $a_2$ in both the first and the last column. 
Hence if $\BWT_K$ were rank-invertible we should have
$$
f_K(\langle 2,1,2i+3\rangle,a,1) = 1,\qquad
f_K(\langle 2,1,2i+3\rangle,a,2) = 2.
$$
At the same time, in the cyclic rotations matrix for $s_2$, $a_1$ precedes $a_2$ in the last columns, but in the first column $a_2$ precedes $a_1$ since the two rotations prefixed by $a$ differ in the $(i+3)$-th column and $c <_\pi b$. Hence we should have
$$
f_K(\langle 2,1,2i+3\rangle,a,1) = 2,\qquad
f_K(\langle 2,1,2i+3\rangle,a,2) = 1
$$
hence $\BWT_K$ cannot be rank-invertible.\qed
\end{proof}

The following lemma shows that $\BWT_K$ is not rank-invertible if $K$ contains anywhere a triplet $(\pid,\pre,\pi)$, with $\pi\neq\pid$.

\begin{lemma}\label{lem:*ir*}
Let $|\Sigma|\geq 2$ and $K=(\pid,\pi_1,\ldots,\pi_{i-1},\pid,\pre,\pi, \pi_{i+3},\ldots,\pi_{k-1})$, $i\geq 0$, with $\pi\neq\pid$. Then $\BWT_K$ is not rank-invertible.
\end{lemma}

\begin{proof}
As in the proof of Lemma~\ref{lem:*ii*}, we can consider the words $s_1=a_1b^ia_2b^{i+1}bb$ and $s_2=a_1b^{i+1}a_2b^{i+1}b$ in case of binary alphabet, and the words $s_1= a_1c^iba_2c^iccc$ and $s_2=a_1c^{i+1}ba_2c^{i+1}c$ in the general case by assuming that there are two symbols, say $b$ and $c$, such that their relative order according to $\pi$ is reversed, that is, $b<c$ and $c <_\pi b$.  Recall that we use subscripts to distinguish the two different occurrences of the symbol $a$.
In the cyclic rotations matrix for $s_1$, in the first column $a_2$ precedes $a_1$ while in the last column $a_1$ precedes $a_2$. At the same time, in both the first and the last column of the cyclic rotations matrix for $s_2$, $a_2$ precedes $a_1$. Reasoning as in the proof of Lemma~\ref{lem:*ii*} we get that $\BWT_K$ cannot be rank-invertible.\qed
\end{proof}

The following lemma shows that $\BWT_K$ is not rank-invertible if $K$ contains anywhere a triplet $(\pre,\pid,\pi)$, with $\pi\neq\pre$.

\begin{lemma}\label{lem:*ri*}
Let $|\Sigma|\geq 2$ and $K=(\pid,\pi_1,\ldots,\pi_{i-1},\pre,\pid,\pi, \pi_{i+3},\ldots,\pi_{k-1})$, $i\geq 0$, with $\pi\neq\pre$. 
Then $\BWT_K$ is not rank-invertible.
\end{lemma}

\begin{proof}
We reason as in the proof of Lemma~\ref{lem:*ir*} considering again the words  $s_1=ab^iab^{i+1}bb$ and $s_2=ab^{i+1}ab^{i+1}b$ in case of binary alphabet and the words $s_1= ac^ibac^iccc$ and $s_2=ac^{i+1}bac^{i+1}c$ in the general case.\qed
\end{proof}

We are now ready to establish the main result of this section.

\begin{theorem}\label{theo:mainRI}
If $|\Sigma|\geq 2$, $\BWT$ and $\ABWT$ are the only transformations $\BWT_K$ which are rank invertible.
\end{theorem}

\begin{proof}
For $|K|=2$, the result follows from Lemma~\ref{lemma:ix}. 
Let us suppose $K=(\pid,\pi_1,\ldots,\pi_{k-1})$ with $k>2$ and assume $\BWT_K$ is rank invertible. Both in the case of binary alphabet and in the general case, by Corollary \ref{cor:ix*}, we must have $\pi_1 = \pid$ or $\pi_1=\pre$. 
If $\pi_1 = \pid$ and $\BWT_K\neq \BWT$ then the $k$-tuple $K$ must contain the triplet $(\pid,\pid,\pi)$ with $\pi\neq \pid$ which is impossible by Lemma \ref{lem:*ii*}. If $\pi_1=\pre$, by Lemma \ref{lem:*ir*} $\pi_2=\pid$. We have therefore established that $K$ has the form $K=(\pid,\pre,\pid, \pi_3, \ldots,\pi_{k-1})$. By Lemma \ref{lem:*ri*} it is $\pi_3=\pre$. By iterating the same reasoning we can conclude that $\BWT_k$ coincides with $\ABWT$, concluding the proof.\qed
\end{proof}

\section{Efficient computation of the ABWT}

The bottleneck for the computation of $\ABWT$ (as well as of any transformation $BWT_K$) of a given string $w$ is the $\palt$-based (the $\preceq_K$-based) sorting of its cyclic rotations. In Remark~\ref{rem:time} we have observed that, 
if a unique end-of-string symbol, which is smaller than any other symbol in the alphabet, is appended to the input string,
all transformations in the class $BWT_K$ can be computed in linear time by first building the suffix tree for the input string. However, for computing the $\BWT$ this strategy has never been used in practice. The reason is that the algorithms for building the suffix tree, although they take linear time, have a large multiplicative constant and are not fast in practice. In addition, the suffix tree itself requires a space of about ten/fifteen times the size of the input which is a huge amount of temporary space that is not necessarily available (considering also that saving space is the primary reason for using the $\BWT$). For the above reasons the $\BWT$ is usually computed by first building the Suffix Array~\cite{Karkkainen:2003,MF02j} which is the array giving the lexicographic order of all the suffixes of the input string.   

A fundamental result on Suffix Array construction is the technique in~\cite{Karkkainen:2006} that, using the concept of  {\em difference cover}, makes it possible to design efficient Suffix Array construction algorithms for different models of computation including RAM, External Memory, and Cache Oblivious. 

In this section, we show that this technique can be adapted to compute the $\ABWT$ within the same time bound of the $\BWT$. 

Firstly, in order to use the notion of suffix array for the computation of $\ABWT$ we need to extend the definition of alternating lexicographic order also for strings having different length. 

\begin{definition}\label{def:lalt}
Let $x=x_0x_1\ldots x_{s-1}$ and $y=y_0y_1\ldots y_{t-1}$ with $s<t$. 
\begin{enumerate}
    \item  If $x$ is not a prefix of $y$ and $i$ is the smallest index in which $x_i\neq y_i$ Then, if $i$  is even $x\lalt y$ iff $x_i<y_i$. Otherwise, if $i$ is odd  $x\lalt y$ iff $x_i>y_i$. 
    \item If $x$ is a prefix of $y$, we say that $x\lalt y$ if $|x|$ is even, $y\lalt x$ if $|x|$ is odd.
\end{enumerate}
\end{definition}

Suffix array algorithms often assume that the input string ends with a unique end-of-string symbol smaller than any other in the alphabet $\Sigma$.
Remark that if we append the end-of-string symbol $\$$ to the string $w$, the $\lalt$-order relation between two suffixes of $w\$$ is determined by using Definition \ref{def:lalt} (case 1). Moreover, using the end-of-string symbol $\$$ implies that the $\palt$-based sorting of the cyclic rotations of input string is induced by the $\lalt$-based sorting of its suffixes. Note that this property does not hold in general. However, it is easy to verify that, apart from the symbol $\$$, the output $abwt(w\$)$ may be different from $abwt(w)$ and the number of equal letter runs can be greater (see Fig. \ref{fig:eofsymbol}). 

Here we assume that the input string $w$ contains a unique end-of-string symbol $\$$, but, in the next section, we show how to remove this hypothesis by using combinatorial properties of some special rotations of the input string. 

To illustrate the idea behind difference cover algorithms, in the following, given a positive integer $v$, we denote by $[0,v)$ the set $\{0, 1, \ldots, v-1\}$.

\begin{definition}\label{def:dc}
A set $D \subseteq [0,v)$ is a {\em difference cover} modulo $v$ if every integer in $[0,v)$ can be expressed as a difference, modulo $v$, of two elements of $D$, i.e.
$$
\{(i-j) \bmod v \mid i,j, \in D\} = [0,v).
$$\qed
\end{definition}


For example, for $v=7$ the set $\{0,1,3\}$ is a difference cover, since $0=0-0$, $1=1-0$, $2=3-1$, $3=3-0$, $4 = 0 - 3 \bmod 7$, and so on. An algorithm by Colbourn and Ling~\cite{ipl/ColbournL00} ensures that for any $v$ a difference cover modulo $v$ of size at most $\sqrt{1.5v} + 6$ can be computed in $O(\sqrt{v})$ time.  The suffix array construction algorithms described in~\cite{Karkkainen:2006} are based on the general strategy shown in Algorithm~\ref{algo:dc}. Steps~3 and~4 rely heavily on the following property of Difference covers: for any $0 \leq i,j< n$ there exists $k<v$ such that 
$(i+k) \bmod v \in D$ and $(j+k) \bmod v \in D$. This implies that to compare lexicographically suffixes $w[i,n-1]$ and $w[j,n-1]$ it suffices to compare at most $v$ symbols since $w[i+k,n-1] $ and $w[j+k,n-1]$ are both sampled suffixes and their relative order has been determined at Step~2.

\begin{algorithm}[t]
\algorithmicrequire{A string $w$ of length $n$ and a modulo-$v$ difference cover~$D$ }\\
\algorithmicensure{$w$'s suffixes in lexicographic order}
\begin{algorithmic}[1]
\State Consider the  $(n|D|)/v$ suffixes $w[i,n-1]$ starting at positions $i$ such that $i\bmod v \in D$. These suffixes are called the {\em sampled suffixes}. 
\State Recursively sort the sampled suffixes
\State Sort non-sampled suffixes
\State Merge sampled and non-sampled suffixes
\end{algorithmic}
\caption{{\sc Difference cover suffix sorting.}\label{algo:dc}}
\end{algorithm}

\def\startpos#1{\stackrel{#1}{ \cdot}\mkern-4mu}

To see how the algorithm works consider for example $v=6$, $D=\{0,1,3\}$ and the string $w=abaacabaacab\$$.  The sampled suffixes are those starting at positions $0,1,3,6,7,9,12$. To sort them, consider the string over $\Sigma^v$ whose elements are the $v$-tuples starting at the sampled positions in the order $0,6,12,1,7,3,9$:
$$
R[0,6] =\; \stackrel{w[0,5]}{abaaca}\;\,
\stackrel{w[6,11]}{baacab}\;\,
\stackrel{w[12,18]}{\$\$\$\$\$\$}\;\, 
\stackrel{w[1,6]}{baacab}\;\,
\stackrel{w[7,12]} {aacab\$}\;\, 
\stackrel{w[3,8]}{acabaa}\;\, 
\stackrel{w[9,14]}{cab\$\$\$}
$$
(note we have added additional \$'s to make sure all blocks contain $v$ symbols). The difference cover algorithm then renames each $v$-tuple with its lexicographic rank. Since
$$
\$\$\$\$\$\$ \plex aacab\$ \plex abaaca \plex  acabaa  
\plex baacab \plex cab\$\$\$
$$
the renamed string is $R_{bwt} = [2, 4, 0, 4, 1, 3, 5]$. The crucial observation is that the suffix array for $R_{bwt}$, which in our example is $SA(R_{bwt}) = [2,4,0,5,1,3,6]$, provides the lexicographic ordering of the sampled suffixes. Indeed $R[2] = w[12,18]$ is the smallest sampled suffix, followed by $R[4]=w[7,12]$, followed by $R[0] = w[0,5]$, and so on. The Suffix Array of $R_{bwt}$ is computed with a recursive call at Step~2, and is later used in Steps~3 and~4 to complete the sorting of all suffixes. 


To compute $\abwt(w)$ with the difference cover algorithm, we consider the same string $R$ but we sort the $v$-tuples according to the {\em alternating} lexicographic order. Since 
$$
\$\$\$\$\$\$ \palt acabaa \palt abaaca \palt aacab\$  
\palt baacab \palt cab\$\$\$
$$
it is $R_{abwt} = [2, 4, 0, 4, 3, 1, 5]$. Next, we compute the Suffix Array of $R_{abwt}$ according to the {\em standard} lexicographic order, $SA(R_{abwt}) = [2,5,0,4,1,3,6]$. We now show that, since $v=6$ is even, $SA(R_{abwt})$ provides the correct {\em alternating} lexicographic order of the sampled suffixes. 

To see this, assume $w[i,n-1]$ and $w[j,n-1]$ are sampled suffixes with a common prefix of length $\ell$. Hence $w[i,i+\ell-1]=w[j,j+\ell-1]$ while $w[i+\ell]\neq w[j+\ell]$. Let $R_{abwt}[t_i]$ and $R_{abwt}[t_j]$ denote the entries in $R_{abwt}$ corresponding to $w[i,i+v-1]$ and $w[j,j+v-1]$. By construction, the suffixes $R_{abwt}[t_i,r]$ and $R_{abwt}[t_j,r]$ have a common prefix of $\lfloor\ell/v\rfloor$ entries (each one corresponding to a length-$v$ block in~$w$) followed respectively by $R_{abwt}[t_i + \lfloor\ell/v\rfloor]$ and $R_{abwt}[t_j+\lfloor\ell/v\rfloor]$ which are different since they correspond to the $v$-tuples $R[t_i + \lfloor\ell/v\rfloor]$ and $R[t_j+\lfloor\ell/v\rfloor]$ which differ since they contain the symbols $w[i+\ell]$ and $w[j+\ell]$ in position $(\ell \bmod v)$. Assuming for example that $w[i+\ell] < w[j+\ell]$, it is $w[i,n-1] \lalt w[j,n-1]$ depending on whether $\ell$ is odd or even. Since $v$ is even, $\ell$ is even iff $\ell \bmod v$ is even, hence
\begin{align*}
w[i,n-1] \lalt w[j,n-1] 
&\Longleftrightarrow  
R[t_i + \lfloor\ell/v\rfloor] \palt R[t_j+\lfloor\ell/v\rfloor] \\
&\Longleftrightarrow 
R_{abwt}[t_i + \lfloor\ell/v\rfloor] < R_{abwt}[t_j+\lfloor\ell/v\rfloor]\\
&\Longleftrightarrow 
R_{abwt}[t_i, r] \plex R_{abwt}[t_j,r]
\end{align*}
which shows that the standard Suffix Array for $R_{abwt}$ provides the alternating lexicographic order of the sampled suffixes, as claimed. 

Summing up, after building the string $R_{abwt}$, 
at Step~2 we compute $SA(R_{abwt})$ using the standard Difference cover algorithm, or any other suffix sorting algorithm. Finally, Step~3 and~4 can be easily adapted to the alternating lexicographic order using its property that for any symbol $c\in\Sigma$ and strings $\alpha, \beta \in\Sigma^*$ it is
\begin{equation}\label{eq:reverse}
c\alpha \lalt c\beta \;\;  \Longleftrightarrow\;\; \beta \lalt \alpha.
\end{equation}
For example, to compare $w[0,12]$ with $w[5,12]$ we notice that after $w[0]=w[5]$ we reach the sampled suffixes $w[1,12]$ and $w[6,12]$ corresponding to $R[3,6]$ and $R[1,6]$. According to $SA(R_{abwt})$ it is $R[1,6] \plex R[3,6]$ which implies $w[6,12] \lalt w[1,12]$, and by~\eqref{eq:reverse} $w[0,12] \plex w[5,12]$. Since from the alternating lexicographic order of $w$'s suffixes $\abwt(w)$ can be computed in linear time, the results in~\cite{Karkkainen:2006} can be translated as follows.

\begin{theorem}
Given a string $w[0,n-1]$ ending with a unique end-of-string symbol, we can compute $\abwt(w)$ in RAM in $\Oh(n)$ time, or in $\Oh(n\log\log n)$ time but using only $n + o(n)$ words of working space. In external memory, using $D$ disks with block size $B$ and a fast memory of size $M$, $\abwt(w)$ can be computed in $\Oh(\frac{n}{DB}\log_{M/B} n/B)$ I/Os and $\Oh(n \log_{M/B} n/B)$ internal work.\qed
\end{theorem}

We point out that the above results cannot be easily extended to the generalized BWTs introduced in Section~\ref{sec:abwt}. The reason is that Step~3 and~4 of the modified Difference cover algorithm hinge on Property~\eqref{eq:reverse} that does not hold in general for the lexicographic orders introduced by Definition~\ref{def_perm}.

\section{Galois words and ABWT computation for arbitrary rotations}

Galois words, introduced in \cite{Reutenauer2006}, are generalization of Lyndon words for the alternating lexicographic order. Roughly speaking, a Galois word is the smallest cyclic rotation within its conjugacy class, with respect to $\palt$ order. Although, in general, Galois and Lyndon words are distinct within a conjugacy class, some properties that hold for Lyndon words are preserved. Some characterizations of Galois words by using infinite words and some properties of words that are obtained as a nonincreasing factorization in Galois words, are studied in~\cite{DOLCE_Reut_Rest2018}. 

In this section, we explore some combinatorial properties of Galois words and, in particular, we show a linear time and space algorithm to find the Galois rotation of a word. These results, on one hand, give an answer to a question posed in~\cite{DOLCE_Reut_Rest2018}. On the other hand, they allow to prove  that, for the computation of the $\ABWT$ of a string $w$, Galois words play a role similar to that of Lyndon words for $\BWT$ \cite{GIA07,BoMaReRoSc_IJFCS_2014}, i.e. the computation of $\ABWT$ can be linearly performed, even if no end-of-string symbol is appended to the input.

\begin{definition}
A primitive word $w$ is a {\em Galois word} if for each nontrivial factorization $w = uv$, one has $w \palt vu$.\qed
\end{definition}

\begin{example}
The words $w=ababba$ and $v=aababb$ are, respectively,  the Galois word and the Lyndon word within the same conjugacy class. Another example is $w=ababaa$ and $v=aaabab$.\qed
\end{example}


Firstly, we show that a string $w$ is a Galois word if it is smaller than its proper suffixes, with respect to $\lalt$ order introduced in Definition \ref{def:lalt} (see Fig. \ref{fig:abwt_suffixes} for an example). 

The following result has been proved in \cite{Reutenauer2006} (Proposition 3.1). Here we report the proof by using our notation.
\begin{lemma}\label{lem: Galois_border}
If a Galois word $w$  has a border, then it has odd length.
\end{lemma}
\begin{proof}
Let $u$ be both suffix and prefix of $w$. This means that $w=uv'=v''u$. By definition, $w=uv'\lalt uv''$.  If $|u|$ would be even, then it should be $v'\lalt v''$. On the other hand, $w=v''u\lalt v'u$ implies that $v''\lalt v'$, a contradiction.\qed
\end{proof}

Lyndon words can be defined as the strings that are smaller of its proper suffixes.
Such a characterization also holds for Galois words, as shown in the following proposition.
A different proof of this result, involving infinite words, is given in~\cite{DOLCE_Reut_Rest2018}.
\begin{proposition}
A primitive word $w$ is a Galois word if and only if $w$ is smaller than any of its suffix, with respect to $\lalt$ order.
\end{proposition}
\begin{proof}
Let $w$ be a Galois word and let $v$ a suffix of $w$. This means that $w=uv$. If $v$ is also prefix of $w$, then by Lemma \ref{lem: Galois_border} $v$ has odd length, i.e. $w\lalt v$. If $v$ is not a prefix of $w$, then there exists $0\leq i<|v|-1$ such that $v_i\neq w_i$.   Since $w$ is a Galois word, $uv\lalt vu$. This fact implies that $w=uv\lalt v$.
Conversely, let $w=uv$. Since $uv\lalt v$, we can distinguish two cases, whether $v$ is prefix of $w$ or not. If $v$ is not a prefix of $w$ then $uv\lalt vu$. If $v$ is a prefix of $w$, then the length of $v$ is odd. Therefore if it would be $vu\lalt uv=vu'$ then $u'\lalt u$ that implies $u'\lalt uv$ that is a contradiction.\qed
\end{proof}

It is known that, when Lyndon words are considered, the lexicographic sorting of its suffixes induces the $\plex$-sorting of the conjugates. Such a result is useful to compute the $\BWT$ of a string without using any end-of-string symbol \cite{GIA07}.  The following proposition shows that this property also holds for Galois words. In fact, the $\palt$-based sorting of the cyclic rotations of a primitive Galois word can be reduced to the $\lalt$-based sorting of its suffixes. An example of this property is reported in Fig. \ref{fig:abwt_suffixes}.

\begin{figure}[ht]
{\small
$$\arraycolsep=2.5pt
\begin{array}{cccccc}
 a & b & a & b & b & a    \\
a & b & b & a & a & b    \\
a & a & b & a & b & b    \\
 b & b & a & a & b & a    \\
 b & a & a & b & a & b    \\
 b & a & b & b & a & a  \\
  &   &    &   &   &   \\
 \multicolumn{6}{c}{M_{\alt}(ababba)}
\end{array}
\qquad\qquad
\begin{array}{cccccc}
 a & b & a & b & b & a    \\
a & b & b & a &  &     \\
a &  &  &  &  &     \\
 b & b & a &  &  &     \\
 b & a &  &  &  &     \\
 b & a & b & b & a &   \\
  &   &    &   &   &   \\
 \multicolumn{6}{c}{Suf(ababba)}
\end{array}\qquad\qquad
\begin{array}{cccccc}
 a & b & a & b & b &     \\
a & b & b &  &  &     \\
a & a & b & a & b & b    \\
 b & b &  &  &  &     \\
 b & a & b & b &  &     \\
 b &  &  &  &  &   \\
  &   &    &   &   &   \\
 \multicolumn{6}{c}{Suf(aababb)}
\end{array}$$
}
\caption{Left: the matrix $M_{\alt}$ of all cyclic rotations of the word $ababba$. Center:  the $\lalt$-sorted suffixes of the Galois word $ababba$. Right: the $\lalt$-sorted suffixes of the Lyndon conjugate $aababb$. The $\lalt$-order of the last two suffixes is different from the $\palt$-order of the correspondent cyclic rotations.} \label{fig:abwt_suffixes}
\end{figure}


\begin{proposition}
Let $w$ be a primitive Galois word and let $u', u'', v', v''$ be factors of $w$ such that $w=u'v'=u''v''$. Then, $v'u'\lalt v''u''\iff v'\lalt v''$. 
\end{proposition}

\begin{proof}
Let us assume that $w'=v'u'\lalt w''=v''u''$. There exists $0\leq i<|w|$ such that $w'_i\neq w''_i$. Firstly, let us assume that $i<\min\{|v'|,|v''|\}$. In this case $v'\lalt v''$.  Let us assume now that $v'$ is a prefix of $v''$, i.e. $v''=v's$, for some non-empty string $s$. If $|v'|$ would be odd, then $v'u'\lalt v''u''=v'su''\Rightarrow su''\lalt u'\Rightarrow su''v'\lalt u'v'$, that is a contradiction. So, $|v'|$ is even and by definition $v'\lalt v''$. Let us consider the case $v''$ is a prefix of $v'$, i.e. $v'=v''t$, for some string $t$. If $|v''|$ would be even then $v''tu'=v'u'\lalt v''u''\Rightarrow tu'\lalt u''\Rightarrow tu'v''\lalt u''v''$, that is a contradiction. So, $|v''|$ is odd then, by definition, $v'\lalt v''$.

Conversely, let us suppose that   $v'\lalt v''$. If neither $v'$ is prefix of $v''$ nor $v''$ is prefix of $v'$, then $v'u'\lalt v''u''$. Let us suppose now that $v'$ is prefix of $v''$, i.e. $v''=v's$ then $v'$ has even length. Since $w$ is a Galois word, $u'v'\lalt su''v'\Rightarrow u'\lalt su'' \Rightarrow v'u'\lalt v'su''=v''u''$. Let us suppose now that $v''$ is prefix of $v'$, i.e. $v'=v''t$ then $v''$ has odd length. The fact that $w$ is a Galois word implies that $u''v''\lalt tu'v''\Rightarrow u''\lalt tu' \Rightarrow v''tu'=v'u'\lalt v''u''$. \qed
\end{proof}

It is known that the unique Lyndon conjugate of a string $w$ is one of the elements in the non-increasing factorization of $ww$ into Lyndon words~\cite{Duval83}. As proved in \cite{DOLCE_Reut_Rest2018}, this strategy does not work for Galois words. Hence, we introduce a new linear time and space algorithm, named {\sc FindGaloisRotation}, to find, for each primitive string $w$ of length $n$, its unique cyclic rotation that is a Galois word. Our algorithm is a variant of the one in~\cite{Booth80,Lothaire:2005} to find the Lyndon conjugate of a given string. The algorithm {\sc FindGaloisRotation} uses a \emph{border array} $B$ of length $n+1$ that stores in each position $j>0$ the length of the border of the $j$-length prefix of $w[k,(k+j-1)\bmod n]$, i.e. the Galois rotation starting at position $k$, and $B[0]=-1$. 

\begin{algorithm}[t]
\caption{{\sc FindGaloisRotation}}
\label{Alg:Galois}
   \algorithmicrequire{A primitive string $w$ of length $n$}\\
  \algorithmicensure{The starting position $0\leq k<n$ of the cyclic rotation of $w$ that is a Galois word}
    \begin{algorithmic}[1] 
            \State $i\gets 0$;  $j\gets 1$; $k\gets 0$;
            \State $B[0]\gets -1$;
            \While{$k+j<2n$} 
           \If {$j\leq n$}  $B[j] \gets i;$ 
	     \EndIf 
		\While {$i\geq  0$ \KwAnd $w[(k+j)\bmod n]\neq w[(k+i)\bmod n]$} 
		\If {$i\bmod2=0$} 
			\If{$w[(k+j)\bmod n]<w[(k+i)\bmod n]$} 
			\State $k\gets k+j-i$; $j\gets i$;
			\EndIf
             \Else 
			\If{$w[(k+j)\bmod n]>w[(k+i)\bmod n]$} 
			\State $k\gets k+j-i$; $j\gets i$;
			\EndIf
		\EndIf
		\State $i\gets B[i]$;
		\EndWhile
		\State $i\gets i+1$; $j\gets j+1$;\label{line:++}
            \EndWhile\label{mainwhile}
            \State \textbf{return} $k$
    \end{algorithmic}
\end{algorithm}

At each iteration of the main {\bf while} loop (lines 3--13), $k$ is the starting position of the current candidate for the smallest cyclic rotation (with respect to $\palt$ order), $w[k,(k+j-1)\bmod n]$ is a Galois word and $B[j]=i$ is the length of its border. This means that $w[k,(k+i-1)\bmod n]=w[(k+j-i),(k+j-1)\bmod n]$. So, the characters $w[(k+i)\bmod n]$ and $w[(k+j)\bmod n]$ are compared. 
If 
those
characters are equal, the length of the border is increased. 
If
they
are distinct, different alphabet orders are used depending on whether $i$ is even or not, and the value of $k$ is consequently updated. Note that, even if $k$ is changed, the computed value $B[j+1]$ is the same and the values $B[i]$, with $i\leq j$, do not need to be re-computed. 

\begin{theorem}
Given a primitive string $w$, its unique cyclic rotation that is a Galois word can be computed in linear time and space. 
\end{theorem}
\begin{proof}
We note that the auxiliary memory consists solely of the border array and that the execution time depends linearly on the number of comparisons between the characters in~$w$. To prove that {\sc FindGaloisRotation} requires at most $4n-3$ comparisons,  we consider the quantity $2(k+j)-i$ and show that it always increases after each comparison between the characters $w[(k+j)\bmod n]$ and $w[(k+i)\bmod n]$.  If the two characters are equal, then both $i$ and $j$ are increased by one at Line~\ref{line:++}. If the two characters are different, then the quantity $k+j$ remains unchanged and the value of $i$ is decreased. Finally, note that if $n\geq 2$, the quantity $2(k+j)-i$ is equal to $2$ for the first comparison and it is at most $2(2n-1)$, so the overall number of comparisons is at most $4n-3$ as claimed.\qed 
\end{proof}

\begin{figure}[ht]
{\small
$$\arraycolsep=2.5pt
\begin{array}{cccccc}
          a &          n &          a &          n &          a &          b \\
          a &          n &          a &          b &          a &          n \\
          a &          b &          a &          n &          a &          n \\
          b &          a &          n &          a &          n &          a \\
          n &          a &          b &          a &          n &          a \\
          n &          a &          n &          a &          b &          a \\
            &            &            &            &            &            \\
          \multicolumn{6}{c}{M_{\alt}(banana)}            \\ 
\end{array}
\qquad\qquad
\begin{array}{ccccccc}
          \$&          b &          a &          n &          a &          n &      a\\
          a &          n &          a &          n &          a &          \$&      b\\
          a &          n &          a &          \$&            b &          a &     n \\
          a &          \$&          b &          a &          n &          a &          n \\
          b &          a &          n &          a &          n &          a &      \$ \\
          n &          a &          \$&         b &          a &          n &          a \\
          n &          a &          n &          a &         \$&         b &          a \\
           &            &            &            &            &            \\
           \multicolumn{6}{c}{M_{\alt}(banana\$)}            \\ 
\end{array}
\qquad\qquad
\begin{array}{ccccccc}
          \$ &         a &          n &          a &          n &          a &          b \\
          a &          n &          a &          n &          a &          b &          \$\\
          a &          n &          a &          b &          \$&          a &          n \\
          a &          b &          \$ &         a &          n &          a &          n \\
          b &          \$ &         a &          n &          a &          n &          a \\
          n &          a &          b &          \$ &         a &          n &          a \\
          n &          a &          n &          a &          b &         \$ &          a \\
            &            &            &            &            &            \\
          \multicolumn{6}{c}{M_{\alt}(ananab\$)}            \\ 
\end{array}
$$
}
\caption{Left: the matrix $M_{\alt}$ of all cyclic rotations of the word $w = banana$, sorted by using $\palt$-order. The output is $abwt(w)=bnnaaa$. Center:  the matrix $M_{\alt}$ of the word $banana\$$. The output is $abwt(banana\$)=abnn\$aa$. Right: the matrix $M_{\alt}$ of the word $ananab\$$, where $ananab$ is the Galois conjugate of $w$. The output is $abwt(ananab\$)=b\$nnaaa$.} \label{fig:eofsymbol}
\end{figure}

The next corollary shows how to use {\sc FindGaloisRotation} procedure to compute in linear time the $\ABWT$ of an input string without using any end-of-string symbol. 

\begin{corollary}
The $\ABWT$ of a generic string $w$ can be computed in linear time.
\end{corollary}

\begin{proof}
We apply {\sc FindGaloisRotation} to $w$ to find its Galois conjugate $w'$. Then we apply to $w'\$$ the algorithm described in previous section. By using Remark \ref{rem:conjugacy}, we can deduce that $\abwt(w)$ can be obtained from $\abwt(w'\$)$ by just removing $\$$ from the output (see Fig. \ref{fig:eofsymbol} for an example).\qed
\end{proof}

In literature some improvements of the algorithms for finding the Lyndon conjugate have been proposed~\cite{Shiloach81}. It is open the question whether similar improvements can be found for Galois words, by reducing the number of comparison or the amount of auxiliary memory used by algorithm {\sc FindGaloisRotation}. Moreover, it would be interesting to investigate whether similar strategies can be applied to other generalized $\BWT$s.

\section*{Acknowledgements}
RG and GM are partially supported by INdAM-GNCS project 2019 \vir{Innovative methods for the solution of medical and biological big data} 
and 
MIUR-PRIN project \vir{Multicriteria Data Structures and Algorithms: from compressed to learned indexes, and beyond} grant n.~2017WR7SHH.

GR and MS are partially supported by the project MIUR-SIR CMACBioSeq \vir{Combinatorial methods for analysis and compression of biological sequences} grant n.~RBSI146R5L. 

\section*{References}

\bibliographystyle{plain}

\bibliography{BWT}  

\begin{thebibliography}{10}

\bibitem{BNtalg14}
D.~Belazzougui and G.~Navarro.
\newblock Optimal lower and upper bounds for representing sequences.
\newblock {\em ACM T. Algorithms}, 11(4):31:1--31:21, 2015.

\bibitem{BoMaReRoSc_IJFCS_2014}
S.~Bonomo, S.~Mantaci, A.~Restivo, G.~Rosone, and M.~Sciortino.
\newblock Sorting conjugates and suffixes of words in a multiset.
\newblock {\em International Journal of Foundations of Computer Science},
  25(08):1161--1175, 2014.

\bibitem{Booth80}
K.~S. Booth.
\newblock Lexicographically least circular substrings.
\newblock {\em Inf. Process. Lett.}, 10(4/5):240--242, 1980.

\bibitem{bwt94}
M.~Burrows and D.~J. Wheeler.
\newblock A block sorting data compression algorithm.
\newblock Technical report, DIGITAL System Research Center, 1994.

\bibitem{dcc/ChapinT98}
B.~Chapin and S.~Tate.
\newblock Higher compression from the burrows-wheeler transform by modified
  sorting.
\newblock In {\em DCC}, page 532. {IEEE} Computer Society, 1998.
\newblock Full version available from
  \url{https://www.uncg.edu/cmp/faculty/srtate/papers/bwtsort.pdf}.

\bibitem{ipl/ColbournL00}
C.~J. Colbourn and A.~C.~H. Ling.
\newblock Quorums from difference covers.
\newblock {\em Inf. Process. Lett.}, 75(1-2):9--12, 2000.

\bibitem{bioinformatics/CoxBJR12}
A.~Cox, M.~Bauer, T.~Jakobi, and G.~Rosone.
\newblock Large-scale compression of genomic sequence databases with the
  {B}urrows-{W}heeler transform.
\newblock {\em Bioinformatics}, 28(11):1415--1419, 2012.

\bibitem{CDP2005}
M.~Crochemore, J.~D\'esarm\'enien, and D.~Perrin.
\newblock A note on the {B}urrows-{W}heeler transformation.
\newblock {\em Theor. Comput. Sci.}, 332:567--572, 2005.

\bibitem{DAYKIN2017}
J.~Daykin, R.~Groult, Y.~Guesnet, T.~Lecroq, A.~Lefebvre, M.~Léonard, and É.
  Prieur-Gaston.
\newblock A survey of string orderings and their application to the
  {Burrows}-{Wheeler} transform.
\newblock {\em Theor. Comput. Sci.}, 2017.

\bibitem{DOLCE_Reut_Rest2018}
F.~Dolce, A.~Restivo, and C.~Reutenauer.
\newblock On generalized {Lyndon} words.
\newblock {\em Theor. Comput. Sci.}, 2018.

\bibitem{Duval83}
J.{-}P. Duval.
\newblock Factorizing words over an ordered alphabet.
\newblock {\em J. Algorithms}, 4(4):363--381, 1983.

\bibitem{cj/fenwick96}
P.~Fenwick.
\newblock The {B}urrows-{W}heeler transform for block sorting text compression:
  Principles and improvements.
\newblock {\em Comput. J.}, 39(9):731--740, 1996.

\bibitem{Ferenczi_Zamboni2012arXiv}
S.~{Ferenczi} and L.~Q. {Zamboni}.
\newblock {Clustering Words and Interval Exchanges}.
\newblock {\em Journal of Integer Sequences}, 16(2):Article 13.2.1, 2013.

\bibitem{FGMS2005}
P.~Ferragina, R.~Giancarlo, G.~Manzini, and M.~Sciortino.
\newblock Boosting textual compression in optimal linear time.
\newblock {\em J. ACM}, 52(4):688--713, 2005.

\bibitem{Ferragina:2000}
P.~Ferragina and G.~Manzini.
\newblock Opportunistic data structures with applications.
\newblock In {\em FOCS 2000}, pages 390--398. IEEE Computer Society, 2000.

\bibitem{Ferragina:2005}
P.~Ferragina and G.~Manzini.
\newblock Indexing compressed text.
\newblock {\em J. ACM}, 52:552--581, 2005.

\bibitem{tcs/GagieMS17}
T.~Gagie, G.~Manzini, and J.~Sir{\'{e}}n.
\newblock Wheeler graphs: {A} framework for {BWT}-based data structures.
\newblock {\em Theor. Comput. Sci.}, 698:67--78, 2017.

\bibitem{GesselRestivoReutenauer2012}
I.~M. Gessel, A.~Restivo, and C.~Reutenauer.
\newblock A bijection between words and multisets of necklaces.
\newblock {\em Eur. J. Combin.}, 33(7):1537 -- 1546, 2012.

\bibitem{GeRe}
I.~M. Gessel and C.~Reutenauer.
\newblock Counting permutations with given cycle structure and descent set.
\newblock {\em J. Comb. Theory A}, 64(2):189--215, 1993.

\bibitem{GMRRS_DLT2018}
R.~Giancarlo, G.~Manzini, A.~Restivo, G.~Rosone, and M.~Sciortino.
\newblock {Block Sorting-Based Transformations on Words: Beyond the Magic BWT}.
\newblock In {\em DLT}, pages 1--17. Springer International Publishing, 2018.

\bibitem{GMRS_CPM2019arxiv}
R.~Giancarlo, G.~Manzini, G.~Rosone, and M.~Sciortino.
\newblock A new class of searchable and provably highly compressible string
  transformations.
\newblock {\em CoRR}, abs/1902.01280, 2019.

\bibitem{GIA07}
R.~Giancarlo, A.~Restivo, and M.~Sciortino.
\newblock From first principles to the {B}urrows and {W}heeler transform and
  beyond, via combinatorial optimization.
\newblock {\em Theor. Comput. Sci.}, 387:236 -- 248, 2007.

\bibitem{Gusfield1997}
D.~Gusfield.
\newblock {\em Algorithms on Strings, Trees, and Sequences - Computer Science
  and Computational Biology}.
\newblock Cambridge University Press, 1997.

\bibitem{Karkkainen:2003}
J.~K\"{a}rkk\"{a}inen and P.~Sanders.
\newblock Simple linear work suffix array construction.
\newblock In {\em Automata, Languages and Programming}, volume 2719 of {\em
  LNCS}, pages 943--955. Springer Berlin Heidelberg, 2003.

\bibitem{Karkkainen:2006}
J.~K\"{a}rkk\"{a}inen, P.~Sanders, and S.~Burkhardt.
\newblock Linear work suffix array construction.
\newblock {\em J. ACM}, 53:918--936, 2006.

\bibitem{KKbioinf15}
K.~Kimura and A.~Koike.
\newblock Ultrafast {SNP} analysis using the {B}urrows-{W}heeler transform of
  short-read data.
\newblock {\em Bioinformatics}, 31(10):1577--1583, 2015.

\bibitem{LiDurbin10}
H.~Li and R.~Durbin.
\newblock Fast and accurate long-read alignment with {B}urrows-{W}heeler
  transform.
\newblock {\em Bioinformatics}, 26(5):589--595, 2010.

\bibitem{Lothaire:2005}
M.~Lothaire.
\newblock {\em Applied Combinatorics on Words (Encyclopedia of Mathematics and
  its Applications)}.
\newblock Cambridge University Press, New York, NY, USA, 2005.

\bibitem{books/MBCT2015}
V.~M{\"{a}}kinen, D.~Belazzougui, F.~Cunial, and A.~I. Tomescu.
\newblock {\em Genome-Scale Algorithm Design: Biological Sequence Analysis in
  the Era of High-Throughput Sequencing}.
\newblock Cambridge University Press, 2015.

\bibitem{MantaciRRS07}
S.~Mantaci, A.~Restivo, G.~Rosone, and M.~Sciortino.
\newblock An extension of the {B}urrows-{W}heeler {T}ransform.
\newblock {\em Theor. Comput. Sci.}, 387(3):298--312, 2007.

\bibitem{MantaciRRS08}
S.~Mantaci, A.~Restivo, G.~Rosone, and M.~Sciortino.
\newblock A new combinatorial approach to sequence comparison.
\newblock {\em Theory Comput. Syst.}, 42(3):411--429, 2008.

\bibitem{MantaciRRS17}
S.~Mantaci, A.~Restivo, G.~Rosone, and M.~Sciortino.
\newblock {Burrows-Wheeler Transform and Run-Length Enconding}.
\newblock In {\em Combinatorics on Words - 11th International Conference,
  {WORDS} 2017. Proceedings}, volume 10432 of {\em LNCS}, pages 228--239.
  Springer, 2017.

\bibitem{MantaciRRSV17}
S.~Mantaci, A.~Restivo, G.~Rosone, M.~Sciortino, and L.~Versari.
\newblock Measuring the clustering effect of {BWT} via {RLE}.
\newblock {\em Theor. Comput. Sci.}, 698:79--87, 2017.

\bibitem{MaReSc}
S.~Mantaci, A.~Restivo, and M.~Sciortino.
\newblock Burrows-{W}heeler transform and {S}turmian words.
\newblock {\em Information Processing Letters}, 86:241--246, 2003.

\bibitem{MantaciRS08}
S.~Mantaci, A.~Restivo, and M.~Sciortino.
\newblock Distance measures for biological sequences: Some recent approaches.
\newblock {\em Int. J. Approx. Reasoning}, 47(1):109--124, 2008.

\bibitem{Manzini2001}
G.~Manzini.
\newblock An analysis of the {B}urrows-{W}heeler transform.
\newblock {\em J. ACM}, 48(3):407--430, 2001.

\bibitem{MF02j}
G.~Manzini and P.~Ferragina.
\newblock Engineering a lightweight suffix array construction algorithm.
\newblock {\em Algorithmica}, 40:33--50, 2004.

\bibitem{books/daglib/0038982}
G.~Navarro.
\newblock {\em Compact Data Structures - {A} Practical Approach}.
\newblock Cambridge University Press, 2016.

\bibitem{PakRedlich2008}
I.~Pak and A.~Redlich.
\newblock Long cycles in abc-permutations.
\newblock {\em Functional Analysis and Other Mathematics}, 2:87--92, 2008.

\bibitem{PPRS2019}
N.~Prezza, N.~Pisanti, M.~Sciortino, and G.~Rosone.
\newblock {SNPs detection by eBWT positional clustering}.
\newblock {\em Algorithms for Molecular Biology}, 14(1):3, 2019.

\bibitem{RestivoRosoneTCS2009}
A.~Restivo and G.~Rosone.
\newblock {B}urrows-{W}heeler transform and palindromic richness.
\newblock {\em Theor. Comput. Sci.}, 410(30-32):3018 -- 3026, 2009.

\bibitem{RestivoRosoneTCS2011}
A.~Restivo and G.~Rosone.
\newblock Balancing and clustering of words in the {B}urrows-{W}heeler
  transform.
\newblock {\em Theor. Comput. Sci.}, 412(27):3019 -- 3032, 2011.

\bibitem{Reutenauer2006}
C.~Reutenauer.
\newblock Mots de {L}yndon g{\'{e}}n{\'{e}}ralis{\'{e}}s 54.
\newblock {\em {S{\'{e}}m. Lothar. Combin.}}, pages 16, B54h, 2006.

\bibitem{RosoneSciortino_CiE2013}
G.~Rosone and M.~Sciortino.
\newblock {The Burrows-Wheeler Transform between Data Compression and
  Combinatorics on Words}.
\newblock In {\em The Nature of Computation. Logic, Algorithms, Applications -
  9th Conference on Computability in Europe, CiE 2013. Proceedings}, volume
  7921 of {\em LNCS}, pages 353--364. Springer, 2013.

\bibitem{Schindler1997}
M.~Schindler.
\newblock A fast block-sorting algorithm for lossless data compression.
\newblock In {\em DCC}, page 469. {IEEE} Computer Society, 1997.

\bibitem{Shiloach81}
Y.~Shiloach.
\newblock Fast canonization of circular strings.
\newblock {\em J. Algorithms}, 2(2):107--121, 1981.

\bibitem{puglisiSimpson}
J.~Simpson and S.~J. Puglisi.
\newblock Words with simple {B}urrows-{W}heeler transforms.
\newblock {\em Electronic Journal of Combinatorics}, 15, article R83, 2008.

\bibitem{YangZhangWang2010}
L.~Yang, X.~Zhang, and T.~Wang.
\newblock The {B}urrows-{W}heeler similarity distribution between biological
  sequences based on {B}urrows-{W}heeler transform.
\newblock {\em Journal of Theoretical Biology}, 262(4):742--749, 2010.

\end{thebibliography}

\end{document}